%% file: ribp.tex
\documentclass{article} 
\pdfoutput=1
\usepackage{graphicx}
\usepackage[lofdepth,lotdepth]{subfig}
\usepackage{amsmath}
\usepackage{amssymb}
\usepackage{algorithmic}
\usepackage{algorithm}
\usepackage{amsthm}
\usepackage[percent]{overpic}
\newtheorem{example}{Example}
\newtheorem{theorem}{Theorem}

\newtheorem{proposition}{Proposition}

\usepackage{color}
\newcommand{\ericx}[1]{}
\newcommand{\eric}[1]{}
\newcommand{\erict}[1]{}

\newcommand*\samethanks[1][\value{footnote}]{\footnotemark[#1]}
\title{Restricting exchangeable nonparametric distributions} 

\author{Sinead Williamson\thanks{Carnegie Mellon University} \and Zoubin Ghahramani\thanks{University of Cambridge} \and Steven N. MacEachern\thanks{Ohio State University} \and Eric P. Xing\samethanks[1]} 

\def\Pois{\mbox{Poisson}}
\def\Bet{\mbox{Beta}}
\def\Bern{\mbox{Bernoulli}}
\def\NB{\mbox{NegBinom}}
\def\PB{\mbox{PoiBin}}
\begin{document} 
 
\maketitle 
 
\begin{abstract} 
Distributions over exchangeable matrices with infinitely many columns, such as the Indian buffet process, are useful in constructing nonparametric latent variable models. However, the distribution implied by such models over the number of features exhibited by each data point may be poorly-suited for many modeling tasks. In this paper, we propose a class of exchangeable nonparametric priors obtained by restricting the domain of existing models. Such models allow us to specify the distribution over the number of features per data point, and can achieve better performance on data sets where the number of features is not well-modeled by the original distribution.
\end{abstract} 
 
\input{intro}
\input{background}
\input{model}
\input{inference_arx}

\input{results}

\input{conclusion}
\bibliographystyle{plain}
\bibliography{ribp}
\end{document}

%% file: intro.tex
\section{Introduction}

The Indian buffet process~(IBP)\cite{Griffiths:Ghahramani:2005} and the related infinite
gamma-Poisson process (iGaP)\cite{Titsias:2007} are distributions over matrices with exchangeable rows and infinitely many columns, only a finite (but random) number of which contain any non-zero entries. Such distributions have proved useful for constructing flexible latent factor models that do not require us to specify the number of latent factors \emph{a priori}. In such models, each column of the random matrix corresponds to a
latent feature, and each row to a data point. The non-zero elements of
a row select the subset of features that contribute to the corresponding
data point.

However, distributions such as the IBP and the iGaP make certain
assumptions about the structure of the data that may be
inappropriate.  Specifically, such priors impose distributions on the number
of data points that exhibit a given feature, and on the number of
features exhibited by a given data point. For example, in the IBP, the number of features exhibited by a data point is marginally Poisson-distributed, and a feature appears in a new data point with probability $m/(N+1)$, where $N$ is the number of previously seen data points, and $m$ is the number of times that feature has appeared. 

These properties may be too constraining for many modeling
tasks. There are a number of cases where we might want to increase the
flexibility of these models by allowing non-Poisson marginals over the number of latent features per data point, or by adding constraints on the number of features. For example, the IBP has been used to select
possible next states in a hidden Markov model \cite{Fox:Sudderth:Jordan:Willsky:2009}. In such a model, we do not expect to see a state that allows \emph{no} transitions (including self-transitions). Nonetheless, because a data point in the IBP can
have zero features with non-zero probability, our prior
supports states with no valid transition distribution. Similarly, the
iGaP has been used to model features in images~\cite{Titsias:2007}, and we may wish to
exclude the possibility of a featureless image.

One interesting example arises when we expect, or desire, the latent features to correspond to interpretable features, or causes, of the data \cite{Wood:Griffiths:Ghahramani:2006}. We might believe that each data point exhibits exactly $K$ features -- corresponding perhaps to speakers in a dialog, members of a team, or alleles in a genotype -- but be agnostic about the total number of features in our dataset. A model that explicitly encodes this prior expectation about the number of features per data point will tend to lead to more interpretable and parsimonious results.


In other situations, we may believe that the number of features per data point follows a distribution other than
that implied by the IBP. For example, it is well known that text and
network data tends to exhibit power-law behavior, suggesting a need
for models that impose heavy-tailed distributions on the number of
features.


In the case of the IBP, two- and three-parameter extensions have been proposed that modify the distribution over the
number of data points that exhibit a feature \cite{Thibaux:Jordan:2007, Ghahramani:Griffiths:Sollich:2007, Teh:Gorur:2009}. 
While these extensions increase flexibility in the distributions over
the number of data points exhibiting each feature, the distribution
over the number of features per data point remains Poisson. As we will see, this is an inherent consequence of the use
of a completely random measure as both prior and likelihood. In this
paper, we consider methods for varying the distribution over the
number of features, by removing the completely random
assumption. 

%% file: background.tex
\section{Exchangeability}
We say a finite sequence $(X_1,\dots,X_N)$ is \emph{exchangeable} (see, for example, \cite{Aldous:1985}) if
its distribution is unchanged under any permutation $\sigma$ of
$\{1,\dots,N\}$. Further, we say that an infinite sequence $X_1,X_2,\dots$ is
\emph{infinitely exchangeable} if all of its finite subsequences are
exchangeable. Such distributions are appropriate when we do not
believe the order in which we see our data is important, or when we do
not have access to all data points.

De Finetti's law tells us that a sequence is exchangeable iff the
observations are i.i.d.\ given some latent distribution. This means that
we can write the probability of any exchangeable sequence as
\begin{equation}
P(X_1=x_1,X_2=x_2,\dots) = \int_\Theta \prod_i \mu_\theta(X_i=x_i|\theta)
\nu(\theta) d\theta \label{eq:deF}
\end{equation}
for some probability distribution $\nu$ over parameter space, and some
parametrised family $\{\mu_\theta\}_{\theta\in\Theta}$ of conditional
probability distributions. 

Throughout this paper, we will use the notation $p(x_1,x_2,\dots) = P(X_1=x_1,X_2=x_2,\dots)$ to represent the joint distribution over an exchangeable sequence $x_1,x_2,\dots$, and $p(x_{n+1}|x_1,\dots,x_n)$ to represent the associated predictive distribution. We will also use the notation $p(x_1,\dots, x_n, \theta):=\prod_{i=1}^n \mu_\theta(X_i=x_i|\theta)\nu(\theta)$ to represent the joint distribution over the observations and the directing measure $\theta$.
In general $\theta$ may be infinite dimensional, which motivates the
close link between the exchangeability assumption and the need for
Bayesian nonparametric models.

\subsection{Distributions over exchangeable matrices}\label{sec:exmat}
The Indian buffet process (IBP)\cite{Griffiths:Ghahramani:2005} is a distribution over binary matrices with
exchangeable rows and infinitely many columns. In the de Finetti
representation, the mixing measure $\nu$ is a beta process, and the
conditional distribution $\mu_\theta$ is a Bernoulli process \cite{Thibaux:Jordan:2007}. The beta
process and the Bernoulli process are both \emph{completely random
  measures} -- distributions over random measures that assign
independent masses to disjoint subsets, that can be written in the
form $\Gamma = \sum_{k=1}^\infty \pi_k\delta_{\theta_k}$ \cite{Kingman:1967}. In the
parametrization of the beta process commonly used for the IBP, the
masses of the atoms $\pi_k$ of a sample from a beta process can be seen as the infinitesimal limit of $\Bet(\alpha dH_0,
1-\alpha dH_0)$ random variables, for some positive scalar $\alpha$ and
CDF $H_0$. The masses of the atoms of a sample
from a Bernoulli process are distributed according to
$\Bern(dG_0)$, for some piecewise-constant function $G_0:\mathcal{X}\rightarrow
[0,1]$ with an at most countable number of jumps. In the context of the IBP, $G_0$ is the cumulative function of the beta-process-distributed measure -- so each atom of the beta process gives the probability for a collection of Bernoulli random variables. We can think of the atoms of the beta process as determining the latent probability for a column of a matrix with
infinitely many columns, and the Bernoulli process as sampling binary
values for the entries of that column of the matrix. The resulting matrix has a finite number of non-zero entries, with the
number of non-zero entries in each row distributed as $\Pois(\alpha)$
and the total number of non-zero columns in $N$ rows distributed as
$\Pois(\alpha H_N)$, where $H_N$ is the $N$th harmonic number. The
number of rows with a non-zero entry for a given column exhibits a
``rich gets richer'' property -- a new row has a one in a given column with probability proportional to the number of times a one has appeared in that column in the preceding rows.

Several models have been formulated that allow us to vary the distribution over
the total number of features and the degree to which features are
shared between data points.  A two-parameter
extension of the IBP \cite{Griffiths:Ghahramani:2011, Thibaux:Jordan:2007} can be obtained by introducing an extra parameter
to the beta process, so that the column probabilities are distributed according to the infinitesimal limit of a $\Bet(c\alpha
dH_0,c(1-\alpha dH_0))$ distribution. The parameter $c$ controls the \emph{degree of sharing} of the features in the resulting IBP: As $c\rightarrow
0$, all data points share the same features, and as $c\rightarrow
\infty$, all data points have disjoint feature sets. A three-parameter
extension \cite{Teh:Gorur:2009} replaces the beta process with a completely random
measure called the stable-beta process, which includes the beta
process as a special case. The resulting IBP exhibits power law
behavior: the total number of features exhibited in a dataset of size
$N$ grows as $O(N^s)$ for some $s>0$, and the number of data points exhibiting
each feature also follows a power law.

A related distribution over exchangeable matrices is the infinite
gamma-Poisson process (iGaP)\cite{Titsias:2007}. Here, the de Finetti mixing measure is the
gamma process, and the family of conditional distributions is given by the
Poisson process. The atoms of the gamma process correspond to the
columns of a matrix, in a manner similar to the beta process in the
IBP. In this case, the atoms determine the mean
value of the column, and the Poisson process populates the column of
the matrix with Poisson random variables with this mean. The result
is a distribution over non-negative integer-valued matrices with
infinitely many columns and exchangeable rows. The sum of each row is
distributed according to a negative binomial
distribution.

%% file: model.tex
\section{Removing the Poisson assumption}\label{sec:model}


In Section~\ref{sec:exmat}, we saw that, while existing methods are
able to alter the degree of sharing of features and the total number
of features in the IBP, they have not been able to remove the Poisson
assumption on the number of features per data point. This is noted by Teh and G\"{o}r\"{u}r \cite{Teh:Gorur:2009}, who point out

\begin{quote}
One aspect of the [three-parameter IBP] which is not power-law is the
number of dishes each customer tries. This is simply $\Pois(\alpha)$
distributed. It seems difficult to obtain power-law behavior in this
aspect within a CRM framework, because of the fundamental role played
by the Poisson process.
\end{quote} 

To elaborate on this, note that, marginally, the distribution over the
value of each element $z_k$ of a row $\mathbf{z}$ of the IBP is given by a Bernoulli distribution. Therefore, by the law of rare events, the
sum $\sum_k z_k$ is distributed according to a Poisson distribution. A similar argument applies to the
infinite gamma-Poisson process. In general, any distribution over
exchangeable random matrices based on a homogeneous CRM will have rows
marginally distributed as i.i.d. random variables. In the case of binary
matrices, these random variables must be Bernoulli, so their sum will
either be Poisson, or infinite. Therefore, in order to circumvent the requirement of a Poisson number
of features in an IBP-like model, we must remove the completely random
assumption on either the de Finetti mixing measure or the family of
conditional distributions.

\subsection{Restricting the family of conditional distributions}\label{sec:rBP}

\ericx{This title and the title of the next section are a little confusing to me. In the CRP or Polya Urn setting, I usually understand the prediction distribution of the next draw given the first $N$ draws as a {\it conditional} distribution, and the joint distribution of all draws as a joint or marginal. Now you call the later a conditional (on a draw from $B$), and the former a predictive, which may cause some confusion. Also I am not complete sure why you call the one "direct", implying the other less "direct". Both seem to be a sum constraint.}

We are familiar with the idea of restricting the support of a
distribution to a measurable subset. For example, a truncated
Gaussian is a Gaussian distribution restricted to a certain contiguous
section of the real line. In general, we can restrict the support of an
arbitrary probability distribution $\mu$ on some space $\Omega$ to a measurable
subset $A\subset\Omega$ in the support of $\mu$ by defining $\mu^{|A}(\cdot) := \mu(\cdot) \mathbb{I}(\cdot\in A) / \mu(A)$, where $\mathbb{I}(\cdot)$ is the indicator function.

\begin{theorem}[Restricted exchangeable distributions]\label{thm:1}
We can always restrict the support of an exchangeable distribution on some space by restricting the family of conditional distributions $\{\mu_{\theta}\}_{\theta\in\Theta}$ introduced in Equation~\ref{eq:deF}, to obtain an exchangeable distribution on the restricted space.
\end{theorem}
\begin{proof} 
Consider an unrestricted exchangeable model with de
Finetti representation $p(x_1,\dots,x_N,\theta) = \prod_{i=1}^N \mu_\theta(X_i=x_i) \nu(\theta)$. Let $p^{|A}$ be the restriction of $p$ such that $X_i \in A,
i=1,2,\dots$, obtained by restricting the family of conditional distributions $\{\mu_\theta\}$ to $\{\mu_\theta^{|A}\}$ as described above. Then
\begin{equation*}
p^{|A}(x_1,\dots,x_N,\theta) = \prod_{i=1}^N \mu_\theta^{|A}(X_i=x_i) \nu(\theta) = \prod_{i=1}^N \frac{\mu_\theta(X_i=x_n)}{\mu_\theta(X_i\in A)} \nu(\theta)\, ,
\end{equation*}
and
\begin{equation}
p^{|A}(x_{N+1}|x_1,\dots,x_N)\propto \int_\Theta
\frac{\prod_{i=1}^{N+1}\mu_\theta(X_{i}=x_i)}{\prod_{i=1}^{N+1}\mu_\theta(X_i\in
  A)}\nu(\theta) d\theta  \label{eq:righteq}
\end{equation}
is an exchangeable sequence by construction, according to de Finetti's
law.
%
\end{proof}
We give two examples based on the IBP.

\begin{example}[Restriction to a fixed number of non-zero entries per row]\label{ex:fixrow}
 Recall that, conditioned on a latent beta
process-distributed measure $B:=\sum_k\pi_k\delta_{\theta_k}$, a sample from the IBP is distributed according to a Bernoulli process. This distribution has support in $\{0,1\}^\infty$.  We can restrict the support of this Bernoulli
process to an arbitrary measurable subset $ A \subset \{0,1\}^\infty$ -- for
example, the set of all vectors $\mathbf{z}\in \{0,1\}^\infty$ such that
$\sum_k z_k = S$ for some integer $S$. The conditional distribution of a matrix
$\mathbf{Z}=\{ \mathbf{z}_1,\dots,\mathbf{z}_N\}$ under such a
  distribution is given by:
\begin{equation}
\mu_B^{|S}(Z=\mathbf{Z}) = \frac{\prod_{i=1}^N
\mu_B(Z_i=\mathbf{z}_i)\mathbb{I}(\sum_k z_{ik} =
S)}{(\mu_B(\sum_k Z_{ik} = S))^N}=  \frac{\prod_{k=1}^\infty
  \pi_k^{m_k}(1-\pi_k)^{N-m_k}}{\PB(S|\{\pi_k\}_{k=1}^\infty)^N}\, ,\label{eq:rBeP}
\end{equation}
where $m_k = \sum_n z_{nk}$ and $\PB(\cdot|\{\pi_k\}_{k=1}^\infty)$ is the infinite limit of the Poisson-binomial distribution \cite{Chen:Liu:1997}, which describes the distribution over the number of successes in a sequence of independent but non-identical Bernoulli trials. The probability of $\mathbf{Z}$ given in Equation~\ref{eq:rBeP} is the infinite limit of the conditional Bernoulli distribution \cite{Chen:Liu:1997}, which describes the distribution of the locations of the successes in such a trial, conditioned on their sum.

\end{example}


\begin{example}[Restriction to a random number of non-zero entries per row]\label{ex:randrow}
Rather than specify the number of non-zero entries in each row a priori, we can allow it to be random, with some arbitrary distribution $f(\cdot)$ over the non-negative integers. A Bernoulli process restricted to have
$f$-marginals can be described as
\begin{equation}
\mu_B^{|f}(\mathbf{Z}) = \prod_{i=1}^N\mu_B^{|S_i}(Z_i = \mathbf{z}_i)f(S_i)
=  \bigg(\prod_{k=1}^\infty
  \pi_k^{m_k}(1-\pi_k)^{N-m_k}\bigg) \cdot 
\prod_{i=1}^N\frac{f(S_i)}{\PB(S_i|\{\pi_k\}_{k=1}^\infty}) \, , \label{eq:fBeP}
\end{equation}
where $S_n = \sum_k z_{nk}$.
Again, if we marginalize over $B$, the resulting distribution is
exchangeable, because mixtures of i.i.d. distributions are i.i.d.
\end{example}

We note that, even if we choose $f$ to be $\Pois(\alpha)$, we will not
recover the IBP. The IBP has $\Pois(\alpha)$ marginals over the number non-zero elements per row, but the
conditional distribution is described by a Poisson-binomial
distribution. The Poisson-restricted IBP, however, will have Poisson
marginal \emph{and} conditional distributions.

We also note that the fixed-row-sum model of Example~\ref{ex:fixrow} can be seen as a special case of the random-distribution model of Example~\ref{ex:randrow}, where the distribution $f$ is degenerate on $S$.

Figure~\ref{fig:prior} shows some examples of samples from the single-parameter IBP, with parameter $\alpha = 5$, with various restrictions applied.

\begin{figure}[ht]
  \centering

    \includegraphics[width=0.8\textwidth]{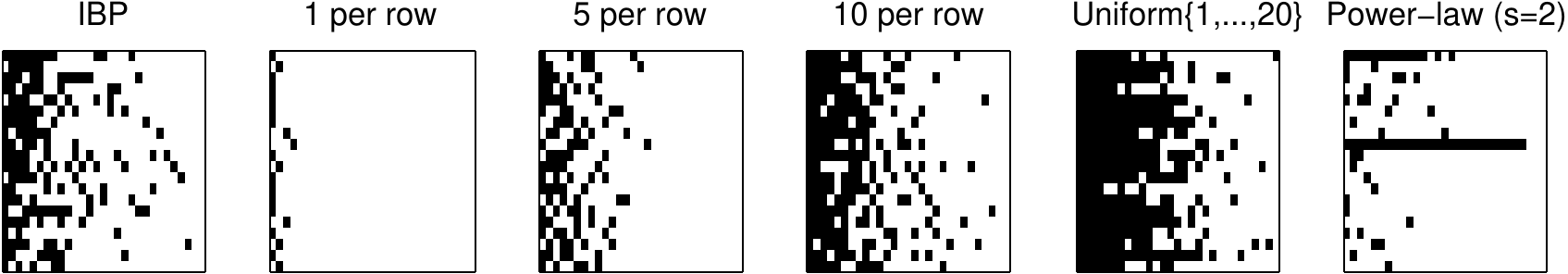}
    \caption{Samples from restricted IBPs.}
    \label{fig:prior}
\end{figure}
\subsection{Direct restriction of the predictive distributions}

The construction in Section~\ref{sec:rBP} is explicitly conditioned on a draw $B$ from the de Finetti mixing measure $\nu$. Since it might be cumbersome to explicitly represent the infinite dimensional object $B$, it is tempting to consider constructions that directly restrict the predictive distribution $p(X_{N+1}|X_1,\dots,X_N)$, where $B$ has been marginalized out. In other words, can we simply sample from an exchangeable distribution
and discard samples that fall outside our region of interest?

We can certainly find examples of exchangeable sequences that remain exchangeable after restricting their conditional distributions:
\begin{example}[Infinite gamma-Poisson process]
Consider restricting the predictive distribution of the infinite gamma-Poisson distribution such that each row sums to $S$.
In the predictive distribution for the iGaP, for each
previously observed feature $k$, we sample an element $X_{nk} \sim
\NB(m_k, n/(n+1))$. We then sample a value $N_n^*\sim
\NB(\theta,n/(n+1))$ and assign $N_n^*$ counts to new features
according to a Chinese restaurant process. If we restrict this model
such that each row sums to $1$, we have:
\begin{equation*}
\begin{split}
p^{|1}(X_{(N+1)k} = 1|X_{1:N}) =& \frac{p(X_{(N+1)k} = 1|X_{1:N})\prod_{j\neq
    k} p(X_{(N+1)j} = 0|X_{1:N})}{p(\sum_jX_{(N+1)j} =
  1|X_{1:N})}\\
=& \begin{cases}
\frac{m_k}{\sum_j m_j + \theta} & \mbox{if feature k has been seen before}\\
\frac{\theta}{\sum_j m_j + \theta} &\mbox{otherwise.}
\end{cases}
\end{split}
\end{equation*}

In other words, the infinite gamma-Poisson process restricted to sum
to one is a Chinese restaurant process. If we restrict the iGaP to
sum to $S$, we have $S$ samples per data point from a Chinese
restaurant process.
\end{example}

However, this result does not hold for direct restriction of arbitrary exchangeable sequences.
\begin{theorem}[Sequences obtained by directly restricting the predictive distribution of an exchangeable sequence are not, in general, exchangeable.]
Let $p$ be the distribution of the unrestricted exchangeable model introduced in the proof of Theorem~\ref{thm:1}. Let $p^{*|A}$ be the distribution obtained by directly restricting
this unrestricted exchangeable model such that $X_n\in A$, i.e.
\begin{equation}
p^{*|A}(x_{N+1}|x_1,\dots,x_N) \propto 
\frac{\int_\Theta \prod_{i=1}^{N+1}\mu_\theta(X=x_i)\nu(\theta) d
  \theta}{\int_\Theta \prod_{i=1}^{N+1}\mu_\theta(X\in
  A) \nu(\theta) d\theta} \, .
\end{equation}
In general, this will not be equal to Equation~\ref{eq:righteq}, and cannot
be expressed as a mixture of i.i.d. distributions. 
\end{theorem}
\begin{proof}
To demonstrate that this is true, consider the counterexample given in Example~\ref{ex:buffet}.
\end{proof}
\begin{example}[A three-urn buffet]
Consider a simple form of the Indian buffet process, with a base measure consisting of three unit-mass atoms. We can represent the predictive distribution of such a model using three indexed urns, each
containing one red ball (representing a one in the resulting matrix) and one blue ball (representing a zero in the resulting matrix). We generate a sequence of
ball sequences by repeatedly picking a ball from each urn, noting the
ordered sequence of colors, and returning the balls to their urns,
plus one ball of each sampled color.\label{ex:buffet} 
\end{example}
\begin{proposition}
The three-urn buffet is exchangeable.
\end{proposition}
\begin{proof}By using the fact that a sequence
is exchangeable iff the predictive distribution given the first $N$
elements of the sequence of the $N+1$st and
$N+2$nd entries is exchangeable \cite{Fortini:Ladelli:Regazzini:2000}, it is trivial to show that this model
is exchangeable and that, for example,
\begin{equation}
\begin{split}
p(&X_{N+1} = (r,b,r), X_{N+2} = (r,r,b)|X_{1:N})\\ &= \frac{m_1 m_2
  (N+1-m_3)}{(N+1)^3}\cdot \frac{(m+1+1)(N+1-m_2)m_3}{(N+2)^3}\\
&=  p(X_{N+1} = (r,r,b), X_{N+2} = (r,b,r)|X_{1:N})\, ,\label{eq:IBP_ex}
\end{split}
\end{equation}
where $m_i$ is the number of times in the first $N$ samples that the
$i$th ball in a sample has been red.
\end{proof} 
\begin{proposition}The directly restricted three-urn scheme (and, by extension, the directly restricted IBP) is not exchangeable.
\end{proposition}
\begin{proof}
Consider the same scheme, but where the outcome is restricted
such that there is one, and only one, red ball per sample. 
The
probability of a sequence in this restricted model is given by
\begin{equation*}
\textstyle p^*(X_{N+1}=x|X_{1:N}) = \sum_{k=1}^3  \frac{m_k}{N+1-m_k} \mathbb{I}(x_i = r) \big/ \sum_{k=1}^3 \frac{m_k}{N+1-m_k}
\end{equation*}
and, for example,
\begin{equation}
\begin{split}
p^*(&X_{N+1} = (r,b,b), X_{N+2} = (b,r,b)|X_{1:N})\\ &=
\frac{\frac{m_1}{N+1-m_1}}{\sum_k
  \frac{m_k}{N+1-m_k}}\cdot
\frac{\frac{m_2}{N+2-m_3}}{\frac{m_2}{N+1-m_2} - \frac{m_2}{N+2-m_2} +
  \sum_k\frac{m_k}{N+1-m_k}}\\
&\neq  p^*(X_{N+1} = (b,r,b), X_{N+2} = (r,b,b)|X_{1:N})\,,\label{eq:rIBP_unex}
\end{split}
\end{equation}
therefore the restricted model is not exchangeable. By introducing a normalizing constant -- corresponding to restricting
over a subset of $\{0,1\}^3$ -- that depends on the previous samples, we
have broken the exchangeability of the sequence.

By extension, a model obtained by directly restricting the predictive distribution of the IBP is not exchangeable.
\end{proof}
 
This section shows that, while directly restricting the predictive distribution of the IBP is appealing because it avoids instantiating the infinite latent measure, this construction \emph{does not} yield an exchangeable distribution. Modifying a Gibbs sampler for the IBP based on the directly restricted predictive distribution would not yield a valid sampler for either the above model, or the exchangeable model described in Section~\ref{sec:rBP}. For the remainder of the paper, we focus on developing valid sampling schemes for the exchangeable model, which we will refer to as a restricted IBP (rIBP).

\ericx{Since you focus on exchangeable construction, I am not sure 3.2 on nonexchangeable is needed, should be much simplified. Currently I do not see a differential treatment of the two alternatives, but then followed by sudden claim that we will do the first one, without a discussion on why it should be favored. In this case, 3.2 become distracting as it take 1.5 pages as if a main model, only ended up not being used without giving a reason.}

%% file: inference_arx.tex
\section{Inference}\label{sec:inference}
In this section, we focus on inference methods for restricted IBPs, since samplers for the restricted iGaP can easily be obtained by modifying existing samplers for the CRP.

We focus on sampling in a truncated model, where we approximate the countably infinite sequence $\{\pi_k\}_{k=1}^\infty$ with a large, but finite, vector $\boldsymbol{\pi}:=(\pi_1,\dots,\pi_K)$, where each atom $\pi_k$ is distributed according to $\mbox{Beta}(\alpha/K,1)$. Conditioned on $\boldsymbol{\pi}$, we can evaluate the probability of a given matrix $\mathbf{Z}$:
\begin{equation}
 \mu_{\mathbf{\pi}}^{|f}(\mathbf{Z})\propto \frac{\prod_{k=1}^K \pi_k^{m_k}(1-\pi_k)^{(N-m_k)}f(S_n)}{\prod_{n=1}^N \PB(S_n|\boldsymbol{\pi})} \label{eq:loglik}
\end{equation}
where $S_n=\sum_k z_{nk}$ and $m_k = \sum_n z_{nk}$. 

Let $g(X|\mathbf{Z})$ be the probability of the data given a binary matrix $\mathbf{Z}$. If the number of entries in each row is random and distributed according to $f$, then we can Gibbs sample each entry of $\mathbf{Z}$ according to
\begin{equation}
\begin{split}
p(z_{nk} &= 1|x_n, \boldsymbol{\pi}, \mathbf{Z}_{\neg nk}, \sum_{j\neq k} z_{nj} = a)\\ &\propto \pi_k\frac{f(a+1)}{p(\sum_k z_k = a+1 |\boldsymbol{\pi})}g(x_n|z_{nk}=1,\mathbf{Z}_{\neg nk},\mathbf{Z}_{\neg n} )\\
p(z_{nk} &= 0|x_n, \boldsymbol{\pi}, \mathbf{Z}_{\neg nk}, \sum_{j\neq k} z_{nj} = a)\\ &\propto (1-\pi_k)\frac{f(a)}{p(\sum_k z_k = a |\boldsymbol{\pi})}g(x_n|z_{nk}=0,\mathbf{Z}_{\neg nk},\mathbf{Z}_{\neg n} )\label{eq:GibbsZ1}
\end{split}
\end{equation}

If the number of non-zero entries per row is fixed, we must resample the location of the non-zero entries. Let $z_n^{(j)}$ indicate the location of the $j$th non-zero entry of $\mathbf{z}_n$. We can Gibbs sample $z_n^{(j)}$ according to

\begin{equation}
p(z_n^{(j)} = k|x_n,\boldsymbol{\pi}, z_m^{(\neg j)})\propto \frac{\pi_k}{1-\pi_k}g(x_n|z_n^{(j)} = k, z_n^{(\neg j)},\mathbf{Z}_{\neg n} )\, .\label{eq:GibbsZ2} 
\end{equation}

Gibbs sampling alone can yield poor mixing, especially in the case where the sum of each row is fixed. To alleviate this problem, we incorporate Metropolis Hastings moves that propose an entire row of $\mathbf{Z}$.

Conditioned on $\mathbf{Z}$, the the distribution of $\boldsymbol{\pi}$ is described by
\begin{equation}
\begin{split}
\nu(\{\pi_k\}_{k=1}^\infty |\mathbf{Z}) &\propto \mu_{\{\pi_k\}}^{|f}(Z=\mathbf{Z}) \nu(\{\pi_k\}_{k=1}^\infty)\\
&= \mu_{\{\pi_k\}}(Z=\mathbf{Z}) \nu(\{\pi_k\}_{k=1}^\infty)\prod_{n=1}^N \frac{f(S_n)}{\mu_{\{\pi_k\}}(|Z|=S_n)}\\
&\propto\frac{\prod_{k=1}^K \pi_k^{(m_k + \frac{\alpha}{K} -1)}(1-\pi_k)^{(N-m_k)}}{\prod_{n=1}^N \PB(S_n|\boldsymbol{\pi})}
\end{split}\label{eq:post}
\end{equation}
The Poisson-binomial term can be calculated exactly in $O(K\sum_kz_{nk})$
using either a recursive algorithm \cite{Barlow:Heidtmann:1984, Chen:Dempster:Liu:1994} or an algorithm based on the characteristic function that uses the Discrete Fourier Transform \cite{Fernandez:Williams:2010}. It can also be approximated using a skewed-normal approximation to the Poisson-binomial distribution \cite{Volkova:1996}. We can therefore sample from the posterior of $\boldsymbol{\pi}$ using Metropolis Hastings steps. Since we believe the posterior will be close to the posterior for the unrestricted model, we use the proposal distribution 
$q(\pi_k|Z) = \mbox{Beta}(\alpha/K + m_k, N+1-m_k)$ to propose new values of $\pi_k$.

In certain cases, we may wish to directly evaluate the predictive distribution $p^{|f}(\mathbf{z}_{N+1}|\mathbf{z}_1,\dots,\mathbf{z}_N)$. Unfortunately, in the case of the IBP, we are unable to perform the integral in Equation~\ref{eq:righteq} analytically. We can, however, \emph{estimate} the predictive distribution using importance sampling. We sample $T$ measures $\boldsymbol{\pi}^{(t)} \sim \nu(\boldsymbol{\pi}|\mathbf{Z})$, where $\nu(\boldsymbol{\pi}|\mathbf{Z})$ is the posterior distribution over $\boldsymbol{\pi}$ in the finite approximation to the IBP, and then weight them to obtain the restricted predictive distribution
\begin{equation}
p^{|f}(\mathbf{z}_{N+1}|\mathbf{z}_1,\dots,\mathbf{z}_N) \approx  \frac{1}{T}\frac{\sum_{t=1}^T w_t \mu_{\boldsymbol{\pi}^{(t)}}^{|f}(\mathbf{z}_{N+1})}{\sum_t w_t} \, ,\label{eq:importance}
\end{equation}
where $w_t = \mu_{\boldsymbol{\pi}^{(t)}}^{|f}(\mathbf{z}_1,\dots, \mathbf{z}_N) / \mu_{\boldsymbol{\pi}^{(t)}}(\mathbf{z}_1,\dots, \mathbf{z}_N)$, and $\mu_{\boldsymbol{\pi}^{(t)}}^{|f}(\cdot)$ is given by Equation~\ref{eq:loglik}

%% file: results.tex
\section{Experimental evaluation}\label{sec:results}
In this paper, we have described how distributions over exchangeable
matrices, such as the IBP, can be modified to allow more flexible 
control on the distributions over the number of latent features, and described methods to perform inference in such models. In this section, we perform experiments on both real and synthetic data. The synthetic data experiments are designed to show that appropriate restriction can yield more interpretable features, and to explore which inference techniques are appropriate in which data regimes. The experiments on real data are designed to show that careful choice of the distribution over the number of latent features in our models can lead to improved predictive performance.

\subsection{Synthetic data}

We begin by evaluating the restricted IBP on synthetic image data. We
generated 50 images, consisting of two binary features selected at
random from a set of four possible features, plus Gaussian
noise. This experiment is a variant of an image
analysis experiment performed in \cite{Griffiths:Ghahramani:2005}.

We tried to learn the latent features using two models: A single-parameter IBP, and a single-parameter IBP restricted to have two features present in each data point. In the restricted model, we alternately sampled $\boldsymbol{\pi}$ and $\mathbf{Z}$ as described in Section~\ref{sec:inference}; for the vanilla IBP we Gibbs sampled the $\pi_k$ in a truncated model. In both cases we fixed $\alpha=2$ and truncated the model to allow $100$ features. Both models were run for $10000$ iterations.

Figure~\ref{fig:toydata} shows the features recovered by both models, and some sample image reconstructions. By incorporating prior knowledge about the number of features, the restricted model is able to find the expected features and achieve superior reconstructions.

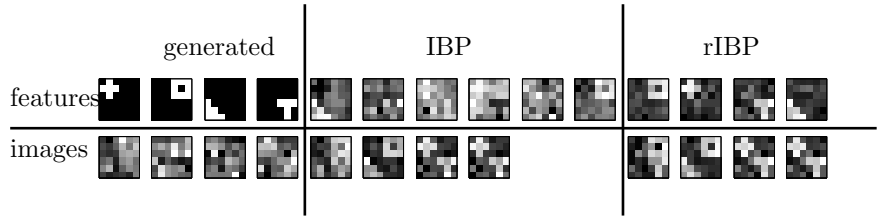
\begin{figure*}[t!]
\begin{center}
\vspace{10pt}
\begin{overpic}[width=0.8\textwidth]{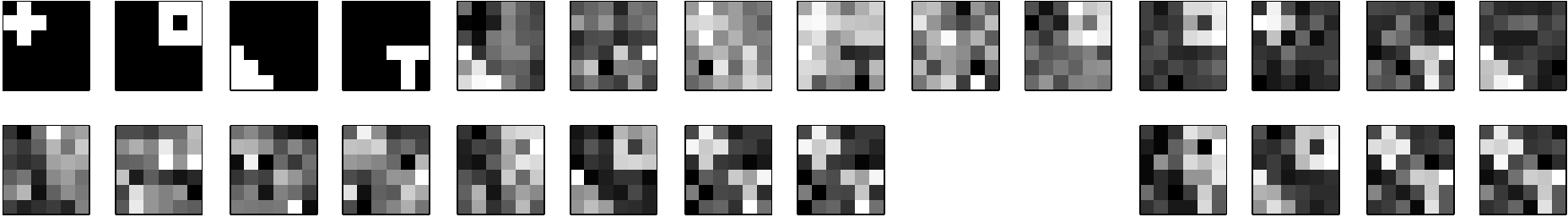}
\put(-12,6.8){\rule{0.95\textwidth}{1pt}}
\put(28.2,-5){\rule{1pt}{80pt}}
\put(71.8,-5){\rule{1pt}{80pt}}
\put(-12,10){features}
\put(-12,3){images}
\put(9,17){generated}
\put(45,17){IBP}
\put(83,17){rIBP}
\end{overpic}
\vspace{8pt}
\caption{Left: Generating features and sample images. Center/right: Features and reconstructions learned using the IBP (center)
  and the IBP restricted to have two features per data point (right).}
\label{fig:toydata}
\end{center}
\end{figure*}

\subsection{Classification of text data}\label{sec:20ng}
The IBP and its extensions have been used to directly model text data\cite{Thibaux:Jordan:2007, Teh:Gorur:2009}. In such settings, the IBP
is used to directly model the presence or absence of words, and so the
matrix is observed rather than latent, and the total number of features is givenby the vocabulary size. We hypothesise that the Poisson assumption made by the IBP is not appropriate for text data, as the statistics of word use in natural language tends to follow a heavier tailed distribution \cite{Zipf:1932}. To test this hypothesis, we modeled a collection of corpora using both an IBP, and an IBP restricted to have heavier tailed distributions over the number of features in each row. Our corpora  were 20 collections of newsgroup postings on various topics (for example, comp.graphics, rec.autos, rec.sport.hockey)\footnote{http://people.csail.mit.edu/jrennie/20Newsgroups/}. To evaluate the quality of the models, we classified held out documents based on their probability under each topic.  This experiment is designed to replicate
an experiment performed by \cite{Teh:Gorur:2009} to compare the
original and three-parameter IBP models.

For our restricted model, we chose a negative Binomial distribution
over the number of words. For both the IBP and the rIBP we estimated the predictive distribution by generating 1000 samples from the posterior of the beta process in the IBP model. No pre-processing of the documents was performed. Since the vocabulary (and hence the feature space) is finite, we used finite versions of both the IBP and the rIBP. Due to the very large state space, we restricted our samples such that, in a single sample, atoms with the same posterior distribution were assigned the same value. In the case of the IBP, we used these samples directly to estimate the predictive distribution; for the restricted model, we used the importance-weighted samples obtained using Equation~\ref{eq:importance}. For each model, $\alpha$ was set to the mean number of features per document in the corresponding group, and the maximum likelihood parameters were used 
for the negative Binomial distribution. For each model, we trained on 1000 randomly selected documents, and tested on a further 1000 documents. 

We evaluated the models by classifying the remaining documents based on
their likelihood under each of the 20 newsgroups. We looked at the fraction correctly classified at $n$ -- ie for each $n=1,\dots,20$ we looked at whether the correct label is one of the $n$ most likely labels. Table~\ref{tab:20ng} shows the fraction of documents correctly classified in the first $n$ labels. The restricted IBP performs uniformly better than the unrestricted model. 

\begin{table}
\begin{center}
\begin{tabular}{|c|c|c|c|c|c|} \hline
 &1 & 2 & 3 & 4 & 5\\
\hline
IBP & 0.591 & 0.726 & 0.796 & 0.848 & 0.878 \\
rIBP & 0.622 & 0.749 & 0.819 & 0.864 & 0.918\\ \hline
\end{tabular}
  \caption{Proportion correct at $n$ on classifying documents from the 20newsgroup dataset.}\label{tab:20ng}
\end{center}
\end{table}

%% file: conclusion.tex
\section{Conclusion}
In this paper we have explored ways of relaxing the distributional assumptions made by existing exchangeable nonparametric processes. The resulting models allow us to specify a distribution over the number of features exhibited by each data point, permitting greater flexibility in model specification. As future work, we intend to explore which applications and models can most benefit from the distributional flexibility afforded by this class of models.

